%% file: main.tex
\pdfoutput=1

\documentclass[numbers]{sigplanconf}

\usepackage[T1]{fontenc}
\usepackage[utf8]{inputenc} 
\usepackage{microtype}
\usepackage{subfigure}
\usepackage{wrapfig}
\usepackage{flushend}
\usepackage{anyfontsize}

\input{prelude}

\title{Coupling Proofs are Probabilistic Product Programs} 

\authorinfo{%
  Gilles Barthe$^1$
  \and
  Benjamin Grégoire$^2$
  \and
  Justin Hsu$^3$
  \and
  Pierre-Yves Strub$^{1,4}$}
{%
  ${}^1$IMDEA Software Institute, Spain
  \and
  ${}^2$Inria, France
  \and
  ${}^3$University of Pennsylvania, USA
  \and
  ${}^4$École Polytechnique, France}
{}

\begin{document}

\toappear{}

\maketitle
\thispagestyle{empty}

\begin{abstract}
\emph{Couplings} are a powerful mathematical tool for reasoning about pairs of
probabilistic processes. Recent developments in formal verification identify a
close connection between couplings and \Sprhl, a relational program logic
motivated by applications to provable security, enabling formal construction of
couplings from the probability theory literature. However, existing work using
\Sprhl merely shows existence of a coupling and does not give a way to prove
quantitative properties about the coupling, needed to reason about
mixing and convergence of probabilistic processes. Furthermore, \Sprhl\ is
inherently incomplete, and is not able to capture some advanced forms of
couplings such as shift couplings. We address both problems as follows.
  
First, we define an extension of \Sprhl, called \Sprprhl, which explicitly
constructs the coupling in a \Sprhl derivation in the form of a probabilistic product program
that simulates two correlated runs of the original program. Existing
verification tools for probabilistic programs can then be directly applied to
the probabilistic product to prove quantitative properties of the coupling.
Second, we equip \Sprprhl\ with a new rule for $\kwhile$ loops, where reasoning
can freely mix synchronized and unsynchronized loop iterations. Our proof
rule can capture examples of \emph{shift couplings}, and the logic is relatively
complete for deterministic programs.

We show soundness of \Sprprhl and use it to analyze two classes of
examples. First, we verify rapid mixing using
different tools from coupling: standard coupling, shift coupling, and
\emph{path coupling}, a compositional principle for combining local
couplings into a global coupling. Second, we verify
(approximate) equivalence between a source and an optimized program
for several instances of loop optimizations from the literature.
\end{abstract}

\category{F.3.1}{Specifying and Verifying and Reasoning about Programs}{}

\keywords{Probabilistic Algorithms, Formal Verification, Relational
  Hoare Logic, Product Programs, Probabilistic Couplings.}

\section{Introduction}
The \emph{coupling method}~\citep{LevinPW09,Lindvall02,Thorisson00,Villani08} is
an elegant mathematical tool for analyzing the relationship between
probabilistic processes.  Informally, couplings correlate outputs of two
probabilistic processes by specifying how corresponding sampling statements are
correlated; reasoning about the correlated processes can then imply interesting
properties of the original processes.

A classical application of couplings is
showing that two probabilistic processes converge in distribution.
Consider, for instance, a symmetric simple random walk over $\mathbb{Z}$:
starting from some initial position $p$, the process repeatedly samples a value
$s$ uniformly in $\{-1,1\}$ and updates its position to $p + s$.  As the process
evolves, the distribution on position spreads out from its initial
position, and converges to a limit distribution that is the same for all initial
positions. This property can be proved by constructing a coupling where the
probability that the two coupled walks end in the same position approaches $1$
as we run the coupled processes.

Beyond merely showing convergence, typically we are interested in \emph{how
  quickly} the processes converge. For instance, a common use of probabilistic
processes is to efficiently sample from a distribution that approximates a
complicated distribution~(e.g., \citet{metropolis1953equation}).  For one
example, the Glauber dynamics~\citep{DBLP:journals/rsa/Jerrum95} approximately
samples from the uniform distribution on proper colorings of a graph---a hard
distribution to compute---by maintaining a coloring and randomly re-coloring a
single vertex at a time. More generally, the Glauber dynamics is an example of
\emph{Markov chain Monte Carlo} (MCMC), a family of techniques that underlie
many computational simulations in science, and machine learning algorithms for
performing probabilistic inference (see \citet{brooks2011handbook} for a
survey). The rate of convergence, measured by the \emph{mixing time} of a
process, determines how many steps we need to run the process before we are
approximately sampling from the target distribution. In this paper, we aim to
formally verify such quantitative features of probabilistic processes.

To this end, recently \citet{BEGHSS15} noticed a close connection between
couplings and probabilistic relational Hoare logic (\Sprhl), originally designed
for reasoning about the computational security of cryptographic
constructions~\citep{BartheGZ09}. Namely, every valid derivation in \Sprhl\
implies the existence of a coupling of the output distributions of the two
programs. Using this observation, \citet{BEGHSS15} verify \Sprhl\ judgments that
imply properties of random walks, the balls-in-bins process, and the birth-death
process. \citet{BGGHS16} extend the connection to approximate couplings in
approximate probabilistic relational Hoare logic (\Saprhl), introduced for
reasoning about differentially private computations~\citep{BartheKOZ12}, and
exploit this connection to prove differential privacy of several examples that
had previously been beyond the reach of formal verification.

While \Sprhl\ is a useful tool for constructing couplings, it has two
limitations. First, it cannot directly reason about the two coupled processes.
This poses a problem for proving
mixing and convergence properties. For instance, \citet{BEGHSS15}
prove that if a certain property $P$ holds on the coupled samples,
then two random walks meet under the coupling. By a
theorem about random walks, this means that the (total variation) distance between the
two distributions is at most the probability of $P$ under the
coupling. However, we do not know what this probability is or how it
grows as we run the random walk for more iterations, since
\Sprhl\ cannot reason about the coupled
process. We run into similar difficulties if we try to prove
convergence using \emph{path coupling}~\citep{bubley1997path}, a
general construction that shows fast convergence by upper-bounding the
expected distance when we make a transition from two coupled
states. While we can express the transition function in code,
\Sprhl\ cannot reason about expectations.
 
Second, \Sprhl\ cannot express some natural classes of couplings because the
rule for $\kwhile$ loops requires that both loops execute in lockstep.
For instance, \emph{shift couplings}~\citep{aldous93shift} allow the two
processes to meet at a random time shift, e.g., the first process at time $t_1$
could track the second process at some other time $t_2 = t_1 + \delta$. This incompleteness of
\Sprhl\ is also a limitation for more standard applications of \Sprhl\ to
program analysis, like validation of loop transformations, and more generally
for any optimization that alters the control flow of the programs.

We address these problems as follows. First, we deepen the connection
between \Sprhl\ and couplings with an observation that is reminiscent
of proofs-as-programs: not only do \Sprhl\ judgements correspond to
couplings, but \Sprhl\ \emph{proofs} encode a \emph{probabilistic
  product program} that constructs the distribution witnessing the
coupling. This program is similar to existing product program
constructions~\citep{ZaksP08,BartheCK11} in that it simulates two
program executions with a single program.  However, the probabilistic
product also coordinates the samplings in the two executions, as
specified by the coupling encoded in the original \Sprhl\ derivation.
Second, we propose a general rule for $\kwhile$ loops; our rule
subsumes several existing rules and is sufficiently expressive to
capture several examples of shift coupling. The resulting system,
which we call \emph{product} \Sprhl or \Sprprhl for short, has several benefits.

First, we obtain a simple, algorithmic procedure to construct the
probabilistic product given a \Sprhl\ derivation. The product directly
simulates the coupled processes, so we can prove quantitative
properties about probabilities or expected values of this coupling
using existing probabilistic verification techniques. Moreover, 
intermediate assertions in a \Sprhl\ derivation can
be soundly transported to the probabilistic product in
\Sprprhl\ derivations.
For instance, we may
prove synchronized control flow in \Sprhl, and directly assume
this fact in the probabilistic product. Many facts are often
practically easier to manipulate in \Sprhl, since \Sprhl\ works purely
on non-probabilistic assertions.

Second, we obtain a powerful logic that can reason about many
examples from the coupling literature, especially shift couplings, and
from the translation validation literature, especially loop
optimizations. On the foundational side, we prove that our logic is
relatively complete for deterministic programs.

\paragraph*{Summary of contributions}
We make the following contributions.
\begin{itemize}
\item We present a proof-relevant program logic \Sprprhl that extracts
  a probabilistic product program from a valid derivation, and prove
  (using the \coq\ proof assistant) that the logic is sound.

\item We propose new rules for loops and random sampling in \Sprprhl.
  We also prove relative completeness for deterministic programs.

\item We demonstrate several applications of \Sprprhl to showing convergence of
  probabilistic processes; several of these examples use path coupling, a
  compositional tool for constructing couplings which bounds convergence in
  terms of expected properties of the coupling, and shift coupling, a
  generalization of coupling where the two processes are allowed to meet at
  different times. We show how to validate some common loop
  transformations.
\end{itemize}

\section{Preliminaries}

We begin by giving a bird's eye view of probabilistic couplings, which take
output distributions from two probabilistic processes and place them in the
same probabilistic space.

\medskip

In the following, we will work with sub-distributions over discrete 
(finite or countable) sets. 

\begin{definition}
  A discrete sub-distribution over a set $A$ is defined by
  a \emph{mass function} $\mu: A \to [0, 1]$ such
  that $\sum_{a \in A} \mu(a)$ is defined and bounded by $1$.
  The quantity $\sum_{a \in A} \mu(a) \in [0, 1]$ is called the
  \emph{weight of $\mu$} and denoted by $\wt{\mu}$.
  The support $\supp(\mu)$ of $\mu$ is defined as
  $\{ x \in A \mid \mu(x) \ne 0 \}$ and is discrete (i.e., countable)
  by construction.
  We denote the set of sub-distributions over $A$ by $\Dist(A)$.
  A \emph{distribution} is a sub-distribution with weight $1$.
  The probability of an event $P$ w.r.t.\ $\mu$, written $\sP \mu P$
  (or $\eP P$ when $\mu$ is clear from the context), is defined as
  $\sum_{\{x \in A \mid P(x)\}} \mu(x)$.
\end{definition}

One can equip distributions with a monadic structure using the Dirac
distributions $\dunit{x}$ for the unit\footnote{%
  $\dunit{x}$ is the distribution where $x$ has probability $1$ and all other
  elements have probability $0$.}
and \emph{distribution expectation} $\dlet x \mu {M(x)}$ for the bind; if $\mu$
is a distribution over $A$ and $M$ has type $A \to \Dist(B)$, then the bind
defines a distribution over $B$:
\begin{equation*}
 \dlet a \mu {M(a)} : b \mapsto \sum_{a} \mu(a) \cdot M(a)(b).
\end{equation*}

When working with sub-distributions over tuples, the probabilistic versions of
the usual projections on tuples are called \emph{marginals}. For distributions
over pairs, we define the \emph{first} and \emph{second marginals}
$\proj_1(\mu)$ and $\proj_2(\mu)$ of a distribution $\mu$ over $A\times B$ by
$\proj_1 (\mu)(a)=\sum_{b\in B} \mu(a,b)$ and $\proj_2 (\mu)(b)=\sum_{a\in A}
\mu(a,b)$.  For a distribution $\mu$ over memories (i.e., maps from variables
$\Vars$ to a set of values $V$) and a set of variables $X \subseteq \Vars$, we define the
$X$-\emph{marginal distribution} of $\mu$ as
$\proj_X(\mu)(m) = \sum
     \{ \mu(m') \mid \forall x \in X .\, m'(x) = m(x) \}$.

We can also construct a sub-distribution over tuples from two distributions: the
\emph{product} sub-distribution of sub-distributions $\mu_1$ and $\mu_2$
of equal weight is defined by
\begin{gather*}
  (\mu_1 \times \mu_2) (a, b) = \mu_1(a) \cdot \mu_2(b).
\end{gather*}
We are now ready to formally define coupling.
\begin{definition}
  Two sub-distributions $\mu_1, \mu_2$ resp. over $A$ and $B$ are said
  to be \emph{coupled} with witness $\mu \in \Dist(A \times B)$,
  written $\coupled{\mu}{\mu_1} {\mu_2}$, iff
  $\proj_1(\mu)=\mu_1$ and $\proj_2(\mu)=\mu_2$.

  For $\post\subseteq A\times B$, we write $\coupledsupp \mu {\mu_1} {\mu_2}
  \post$ if $\coupled \mu {\mu_1} {\mu_2}$ and moreover $\supp
  (\mu)\subseteq\post$. We will sometimes abuse notation and call $\mu$ a
  \emph{coupling} of $\mu_1$ and $\mu_2$.
\end{definition}
As an example, suppose that $A = B$ and $\mu_1 = \mu_2$ are the uniform
distributions. Then, any bijection $f : A \to A$ gives a coupling of $\mu_1$ and
$\mu_2$; we call the resulting coupling $\CpDist{f}{A}$ so that
$\coupled{\CpDist{f}{A}}{\mu_1}{\mu_2}$. 
The coupling $\CpDist{f}{A}$ assigns positive probability only to
pairs $(v,f~v)$ where $v \in A$.

For another example, suppose again that $A = B$ and $\mu_1 = \mu_2$, but the
distributions are not necessarily uniform. Then, the identity map $\id : A \to A$
always gives a coupling of $\mu_1$ and $\mu_2$, correlating samples from both
distributions to be the same. We will write $\CpId{\mu_1}$ for the resulting
coupling, so that $\coupled{\CpId{\mu_1}}{\mu_1}{\mu_2}$.
If $\mu_1$ is the uniform distribution we will sometime write
$\CpId{A}$ instead of $\CpId{\mu_1}$.
Note that $\CpId{\mu_1}$ assigns positive probability only to pairs $(v, v)$
with $\mu_1(v), \mu_2(v) \neq 0$.

To reason about convergence of probabilistic processes, we will use the
total variation distance on distributions (also known as statistical distance).

\begin{definition}
  Let $\mu_1$ and $\mu_2$ be sub-distributions over a countable set
  $A$. The \emph{total variation ($\TV$) distance} between $\mu_1$ and
  $\mu_2$ is defined by:
  \begin{align*}
    \TV(\mu_1, \mu_2) =
      \frac{1}{2} \sum_{a \in A} |\mu_1(a) - \mu_2(a) |.
  \end{align*}
\end{definition}

To bound this distance, it is enough to find a coupling and bound the
probability that the two coupled variables differ; this is the fundamental idea
underlying the coupling method.

\begin{theorem}[Fundamental theorem of couplings (e.g., \citep{Lindvall02})]
  \label{thm:tv}
  Let $\mu_1$ and $\mu_2$ be distributions over a countable set. Then for any
  coupling $\mu$ of $\mu_1, \mu_2$, we have
  \begin{align*}
    \TV(\mu_1, \mu_2) \leq \P {(x, x')} {\mu} {x \neq x'}.
  \end{align*}
\end{theorem}


\section{Product Programs}
\subsection{Language}

We will work with a core, probabilistic imperative language with a command for
random sampling from primitive distributions. The set of commands is defined as
follows:

\begin{tightcenter}
$\begin{array}{r@{\ \ }l@{\quad}l}
\cmd ::= & \iskip                     & \mbox{noop} \\
     \mid& \iabort                    & \mbox{abort} \\
     \mid& \iass{\Var}{\Expr}         & \mbox{deterministic assignment}\\
     \mid& \irnd{\Var}{\DExpr(\Expr)} & \mbox{probabilistic assignment}\\
     \mid& \iseq{\cmd}{\cmd}          & \mbox{sequencing}\\
     \mid& \icond{\Expr}{\cmd}{\cmd}  & \mbox{conditional}\\
     \mid& \iwhile{\Expr}{\cmd}       & \mbox{while loop} \\
\end{array}$
\end{tightcenter}

Here, $\Var$ is a set of \emph{variables}, $\Expr$ is a set of
(deterministic) \emph{expressions}, and $\DExpr$ is a set of
\emph{distribution expressions}. Variables and expressions are typed,
ranging over booleans, integers, lists, etc.  The expression grammar
is entirely standard, and we omit it. We will use metavariables $c$ to represent
commands, $e$ to represents expressions, and $g$ to represent distribution
expressions.

We will use several shorthands for commands:$ \icondT{e}{c}$ for
$\icond{e}{c}{\iskip}$, $\kabort$ for the looping command
$\iwhile{\top}{\iskip}$, and $\cfor{c}{e,k}$ for the $k$-fold composition of
$c$ restricted to $e$. Formally, 
$$\cfor{c}{e,k}\triangleq \iass{\imath}{0};\iwhile{(\imath<k) \land e}{c;\imath\kinc}$$ 
where $\imath$ is a fresh variable. We allow $k$ to be an arbitrary expression;
when $k$ is a constant,
$$\vdash \cfor{c}{e,k} \equiv \overbrace{\SIfT{e}{c};\ldots;\SIfT{e}{c}}^{k~\mbox{times}}$$
for the simple notion of program equivalence introduced in \Cref{fig:equiv}. 

\medskip

The denotational semantics of programs is adapted from the seminal
work of \citet{Kozen79}.
We first interpret every ground type $T$ as a set $\sem{T}{}$; other
constructors $C$ are interpreted as functions $\sem{C}{}$ that respect
their arities.
The semantics of expressions and distribution expressions, denoted $\sem{e}{m}$
and $\sem{g}{m}$ respectively, are parameterized by a \emph{state} $m$ (also
called a \emph{memory}), and are defined in the usual way.
Finally, commands are interpreted as a function from memories to
sub-distributions over memories, where memories are finite maps from
variables to values. Formally, we let $\Mem_X$ denote the set of
memories over the finite set of variables $X$. Moreover, we use
$\uplus$ to denote the standard disjoint union on finite maps, so
$\uplus:\Mem_{X_1} \times \Mem_{X_2} \rightarrow \Mem_{X_1\cup X_2} $
for disjoint finite sets $X_1$ and $X_2$.
The interpretation of $c$, written $\sem{c}{}$, is a function from
$\Mem_X$ to $\Dist(\Mem_X)$, where $\Vars(c)=X$. The definition of
$\sem{\cdot}{}$ is given in \Cref{fig:sem}; note that $c^{b, k}$ is slightly
different than $c_{b, k}$, as the former unrolls the loop $k$ times, while the
latter also drops (via the $\iabort$ instruction) executions that do not exit the
loop after $k$ iterations.

\begin{figure*}
\begin{align*}
  \sem{\iskip}{m} &= \dunit{m} &
  \sem{\iabort}{m} &= \dnull \\
  \sem{\iass x e}{m} &= \dunit{m\subst{x}{\sem{e}{m}}} &
  \sem{\irnd x g}{m} &= \dlet v {\sem{g}{m}} {\dunit{m[x:=v]}} \\
  \sem{\iseq {c_1} {c_2}}{m} &= \dlet {m'} {\sem{c_1}{m}} {\sem{c_2}{m'}} &
  \sem{\icond{e}{s_1}{s_2}}{m} &=
    \text{if $\sem{e}{m}$ then $\sem{c_1}{m}$ else $\sem{c_2}{m}$} \\
  \sem{\iwhile{b}{c}}{m} &=
    \lim_{n \to \infty}\ \sem{\itfor{c}{b,n}}{m} 
 \mathrlap{\qquad\mbox{where }\itfor{c}{b,n} \triangleq
        \overbrace{\icondT{b}{c};\ldots;\icondT{b}{c}}^{n~\mbox{times}}; 
        \icondT{b}{\iabort}} &\\
\end{align*} 

\caption{\label{fig:sem} Interpretation of commands}
\end{figure*}

Last, for any predicate over memories $\Phi$, we say that a command is
$\Phi$-\emph{lossless} iff for any memory $m$ s.t. $\Phi(m)$, the weight of
$\sem{c}{m}$ is equal to $1$; intuitively, such a command terminates with
probability $1$ from every initial memory satisfying $\Phi$. A command is
\emph{lossless} if it is $\True$-lossless, i.e.\ if $\wt{\sem{c}{m}} = 1$ for
every memory $m$.

\subsection{Proof System}
Our proof system manipulates judgments of the form:
$$\Equivw{c_1}{c_2}{\pre}{\post}{c}$$
where $c_1$, $c_2$ are statements over disjoint variables and
$\pre$ and $\post$ are assertions over the variables of $c_1$ and
$c_2$. We will not set a specific syntax for assertions, but one
natural choice is first-order formulas over the program variables of both
programs.

Informally, the judgment is valid if $c$ is a probabilistic product
program for $c_1$ and $c_2$ under the pre-condition $\pre$, i.e.\, for
every initial memory $m=m_1\uplus m_2$ satisfying the pre-condition
$\pre$, the sub-distribution $\sem{c}{m}$ is a coupling for
$\sem{c_1}{m_1}$ and $\sem{c_2}{m_2}$, and moreover
$\supp(\sem{c}{m})$ only contains states that satisfy $\post$.

\begin{definition}[Valid judgment]\label{def:valid}\mbox{}
  \begin{itemize}
  \item
  Two commands $c_1$ and $c_2$ are \emph{separable} iff their
  variables are disjoint, i.e.\, $\Vars(c_1) \cap
  \Vars(c_2) = \emptyset$.

  \item
  A judgment $\Equivw{c_1}{c_2}{\pre}{\post}{c}$ is \emph{valid} iff
  the commands $c_1$ and $c_2$ are separable and for every memory
  $m=m_1\uplus m_2$ such that $m\vDash\pre$, we have
  \[
    \coupledsupp{\sem{c}{m}}{(\sem{c_1}{m_1})}{(\sem{c_2}{m_2})}{\post} .
  \]
  \end{itemize}
\end{definition}
For comparison, judgments in $\Sprhl$ are of the form
$$\Equiv{c_1}{c_2}{\pre}{\post}$$
and assert that for every initial memory $m=m_1\uplus m_2$ such that
$m\vDash\pre$, there exists a coupling $\mu$ such that
\[ \coupledsupp{\mu}{(\sem{c_1}{m_1})}{(\sem{c_2}{m_2})}{\post} .\]
In contrast, our notion of judgment is \emph{proof-relevant}, since
the derivation of the judgment guides the product construction.  We
briefly comment on several rules in our system, presented in
\cref{fig:rules}.  The [\textsc{Conseq}] rule is similar to the rule
of consequence, and reflects that validity is preserved by weakening
the post-condition and strengthening the pre-condition; none of the
programs is modified in this case.

The [\textsc{False}] rule corresponds to the fact that under
a false precondition nothing needs to be proved.

The [\textsc{Case}] rule allows proving a judgment by case analysis;
specifically, the validity of a judgment with pre-condition $\pre$ can
be established from the validity of two judgments, one where the
pre-condition is strengthened with $e$ and the other where the
pre-condition is strengthened with $\neg e$. The
product program is of the form $\icond{e}{c}{c'}$, where $c$ and $c'$
correspond to the product programs built from the proofs
of the first and second premises respectively.

The [\textsc{Struct}] rule allows replacing programs by provably equivalent
programs. The rules for proving program equivalence are given in
\Cref{fig:equiv}, and manipulate judgments of the form $\eqsem{\pre}{c}{c'}$,
where $\pre$ is a relational assertion. We only introduce equations that are
needed for recovering derived rules, striving to keep the notion of structural
equivalence as simple as possible.

The [\textsc{Assg}] rule corresponds to the \Sprhl\ rule for
assignments; in this case, the product program is simply the
sequential composition of the two assignments.

The [\textsc{Rand}] rule informally takes a coupling between the two
distributions used for sampling in the left and right program, and produces a
product program that draws the pair of correlated samples from the coupling.
Since our language supports sampling from distribution expressions and not only
distributions, the rule asks for the existence of a coupling for each
interpretation of the distribution expressions under a valuation satisfying the
pre-condition $\pre$ of the judgment. Furthermore, note that the rule requires
that every element in the support of the coupling
validates the post-condition. This is similar to \Sprhl, and is a natural
requirement given the notion of valid judgment. However, the rule is strictly
more general than the corresponding \Sprhl\ rule, which is
restricted to the case where $\post$ is the graph of some bijection.

The [\textsc{Seq}] rule for sequential composition simply composes the two
product programs in sequence. This rule reflects the compositional property of
couplings.

The [\textsc{While}] rule for \kwhile\ loops constructs a product
program that interleaves synchronous and asynchronous executions of
the loop bodies. The first premise establishes that $k_1$ and $k_2$
are strictly positive when the loop invariant $\Psi$ holds. Then, we
specify an expression $e$---which may mention variables from both
sides---that holds true exactly when at least one of the guards is
true.

Next, the notation $\oplus\{ p_0, p_1, p_2 \}$ indicates that exactly one of the
tests $p_0$, $p_1$, and $p_2$ must hold.  These predicates guide the
construction of the product program. If $p_0$ holds, then both guards should be
equal and we can execute the two sides $k_1$ and $k_2$ iterations respectively,
preserving the loop invariant $\Psi$. If $p_1$ holds and the right loop has not
terminated yet, then the left loop also has not terminated yet (i.e., $e_2$
holds), we may execute the left loop one iteration. If $p_2$ holds and the left
loop has not terminated yet (i.e., $e_1$ holds), then the right loop also has
not terminated yet and we may execute the right loop one iteration.

The last pair of premises deal with termination. Note that
some condition on termination is needed for soundness of the logic:
if the left loop terminates with probability $1$ while the right loop
terminates with probability $0$ (i.e., never), it is impossible to construct a
valid product program since there is no distribution on pairs that has first
marginal with weight $1$ and second marginal with weight $0$. So,
we require that the first and second loops are lossless assuming $p_1$
and $p_2$ respectively. This ensures that with probability $1$, there are only
finitely many steps where we execute the left or right loop separately. Note
that the product program may still fail to terminate with positive probability since
there may be infinite sequences of iterations where $p_0$ holds so both
loops advance, but both programs will yield sub-distributions with the
\emph{same} weight so we can still find a coupling.

With these premises, the construction of the product program is straightforward.
At the top level, the product continues as long as $e$ is true, i.e., while at least one of the two
loops can make progress. Each iteration, it performs a case analysis on the three
predicates. If $p_0$ holds, then we execute the product from running the two
loops up to $k_1$ and $k_2$ iterations respectively. If $p_1$ holds, then we execute
the product from running the left loop one iteration; if $p_2$ then we execute
the product from running the right loop one iteration.

\begin{figure}
\begin{mathpar}
\xinferrule{~}{\eqsem{\pre}{c}{c}}

\xinferrule{\eqsem{\pre}{c_1}{c_2}}{\eqsem{\pre}{c_2}{c_1}}

\xinferrule{~}{\eqsem{\pre}{\irnd{x}{\dunit{x}}}{\iskip}}


\xinferrule{\pre \implies x = e}{\eqsem{\pre}{\iass{x}{e}}{\iskip}}

\xinferrule{~}{\eqsem{\pre}{c;\iskip}{c}}

\xinferrule{~}{\eqsem{\pre}{\iskip;c}{c}}

\xinferrule{\eqsem{\pre}{c_1}{c'_1}}
           {\eqsem{\pre}{c_1;c_2}{c'_1;c_2}}

\xinferrule{\eqsem{\top}{c_2}{c'_2}}
           {\eqsem{\pre}{c_1;c_2}{c_1;c'_2}}

\xinferrule{\pre \implies e}{\eqsem{\pre}{\icond{e}{c}{c'}}{c}}

\xinferrule{\pre \implies \neg e}{\eqsem{\pre}{\icond{e}{c}{c'}}{c'}}

\xinferrule{\eqsem{\pre\land e}{c_1}{c_2} \\
            \eqsem{\pre\land \neg e}{c'_1}{c'_2}
}{\eqsem{\pre}{\icond{e}{c_1}{c'_1}}{\icond{e}{c_2}{c'_2}}}

\xinferrule{\eqsem{e}{c}{c'}}{\eqsem{\pre}{\iwhile{e}{c}}{\iwhile{e}{c'}}}

\xinferrule{~}{\eqsem{\pre}{\iwhile{e}{c}}{\icondT{e}{(c;\iwhile{e}{c})}}}

\end{mathpar}

\caption{Equivalence rules}\label{fig:equiv}
\end{figure}

\begin{figure*}
  $$\begin{array}{c}
\xinferrule*[left=\textsc{Conseq}]
    { \Equivw{c_1}{c_2}{\pre}{\post}{c} \\ \pre'\Longrightarrow\pre \\
      \post\Longrightarrow\post'}
 { \Equivw{c_1}{c_2}{\pre'}{\post'}{c}}
\qquad
\xinferrule*[left=\textsc{False}]
    {~}
 { \Equivw{c_1}{c_2}{\bot}{\post}{\kskip}}
\\[2ex]
 \xinferrule*[left=\textsc{Struct}]{\Equivw{c_1}{c_2}{\pre}{\post}{c} \\
   \eqsem{\pre}{c_1}{c_1'} \\ \eqsem{\pre}{c_2}{c_2'} \\ \eqsem{\pre}{c}{c'}}
 {\Equivw{c_1'}{c_2'}{\pre}{\post}{c'}}
\\[2ex]
\xinferrule*[left=\textsc{Case}]{\Equivw{c_1}{c_2}{\pre\land e}{\post}{c} \\
 \Equivw{c_1}{c_2}{\pre\land \neg e}{\post}{c'}}
 { \Equivw{c_1}{c_2}{\pre}{\post}{\icond{e}{c}{c'}}}
\\[2ex]
  \xinferrule*[left=\textsc{Assg}]
  {~}
  {\Equivw
      {\iass{x_1}{e_1}}{\iass{x_2}{e_2}}
      {\post[e_1/x_1,e_2/x_2]}{\post}
      {\iass{x_1}{e_1};\iass{x_2}{e_2}}}
\\[2ex]
\xinferrule*[left=\textsc{Rand}]
            {\forall m. m\vDash \pre \Longrightarrow
              \coupledsupp{\sem{g}{m}}{\sem{g_1}{m}}{\sem{g_2}{m}}{
                \{ (v_1,v_2) \mid m[x_1:=v_1,x_2:=v_2] \vDash \post \}}}
  {\Equivw
      {\irnd{x_1}{g_1}}{\irnd{x_2}{g_2}}
      {\pre}
      {\post}{\irnd{(x_1,x_2)}{g}}}
\\[2ex]
\xinferrule*[left=\textsc{Seq}]{\Equivw{c_1}{c_2}{\pre}{\Xi}{c} \\
 \Equivw{c'_1}{c'_2}{\Xi}{\post}{c'}}
 { \Equivw{c_1;c'_1}{c_2;c'_2}{\pre}{\post}{c;c'}}

\\[2ex]

  \xinferrule*[left=\textsc{While}]
              {\post \implies k_1>0 \land k_2 >0
                \\
                \post \implies (e_1\lor e_2)=e
                \\
                \post \land e \implies \oplus \{p_0, p_1, p_2\}
                \\\\
      \post \land p_0 \land e \implies e_1= e_2 \\
      \post \land p_1 \land e \implies e_1 \\
      \post \land p_2 \land e \implies e_2 \\\\
      \iwhile{e_1 \land p_1}{c_1} \quad \post\text{-lossless} \\
      \iwhile{e_2 \land p_2}{c_2} \quad \post\text{-lossless} \\\\
     \Equivw{\cfor{c_1}{e_1,k_1}}{\cfor{c_2}{e_2,k_2}}{\post \land p_0}{\post}{c_0} \\
     \Equivw{c_1}{\iskip}{\post \land e_1\land p_1}{\post}{c_1} \\
     \Equivw{\iskip}{c_2}{\post \land e_2\land p_2}{\post}{c_2}}
  {
    \Equivw{\iwhile{e_1}{c_1}}{\iwhile{e_2}{c_2}}
      {\post}
      {\post \land \neg e_1 \land \neg e_2}
      {\iwhile{e}{\icond{p_0}{c_0}{\icond{p_1}{c_1}{c_2}}}}}
\end{array}$$

\caption{\label{fig:rules} Proof rules}
\end{figure*}


\begin{figure*}
$$\begin{array}{c}
\xinferrule*[left=\textsc{Assg-L}]
  {~}
  {\Equivw
      {\iass{x_1}{e_1}}{\iskip}
      {\post[e_1/x_1]}{\post}
      {\iass{x_1}{e_1}}}
\\[2ex]
\xinferrule*[left=\textsc{Rand-L}]
  {\mu_1~\mathrm{lossless}}
  {\Equivw
      {\irnd{x_1}{g_1}}{\iskip}
      {\forall v_1\in\supp(g_1), \post [v_1/x_1]}
      {\post}{\irnd{x_1}{g_1}}}
\\[2ex]
  
\xinferrule*[left=\textsc{Cond-L}]
  {\Equivw{c_1}{c_2}{\pre \land e_1}{\post}{c} \\
   \Equivw{c_1'}{c_2}{\pre \land \neg e_1}{\post}{c'}}
 {\Equivw
     {\icond{e_1}{c_1}{c_1'}}{c_2}
     {\pre}{\post}{\icond{e_1}{c}{c'}}}
 \\[2ex]
 
\xinferrule*[left=\textsc{Cond-S}]
  {\pre \implies e_1 = e_2 \\
   \Equivw{c_1}{c_2}{\pre \land e_1}{\post}{c} \\
   \Equivw{c_1'}{c_2'}{\pre \land \neg e_1}{\post}{c'}}
 {\Equivw
     {\icond{e_1}{c_1}{c_1'}}{\icond{e_2}{c_2}{c_2'}}
     {\pre}{\post}{\icond{e_1}{c}{c'}}}
 \\[2ex]
\xinferrule*[left=\textsc{While-L}]
  {
    \Equivw{c_1}{\iskip}{\post \land e_1}{\post}{c_1} \\
    \iwhile{e_1}{c_1} \quad \post\text{-lossless}}
 {\Equivw{\iwhile{e_1}{c_1}}{\iskip}
     {\post}{\post\land\neg e_1}{\iwhile{e_1}{c_1}}}
\\[2ex]
 \xinferrule*[left=\textsc{While-S}]
  {\post \implies e_1 = e_2 \\
   \Equivw{c_1}{c_2}{\post \land e_1}{\post}{c}}
 {\Equivw
     {\iwhile{e_1}{c_1}}{\iwhile{e_2}{c_2}}
     {\post}{\post \land \neg e_1}{\iwhile{e_1}{c}}} 
  \end{array}$$
\caption{Derived rules}\label{fig:derived}
\end{figure*}

\subsubsection*{Derived rules}
Presentations of relational Hoare logic often include \emph{one-sided
  rules}, which are based on the analysis of a single program (rather
than two).  By reasoning about program equivalence, we can derive all
one-sided rules and the two-sided rules of conditionals and loops from
\Sprhl, presented in \cref{fig:derived}, within our system.\footnote{%
  In fact, the one-sided and two-sided rules are inter-derivable for
  all constructions except random assignments and loops.}  
\begin{proposition}
All the rules in \Cref{fig:derived} (and their symmetric version) are
derived rules.
\end{proposition}
We briefly comment on some of the derived rules. The [\textsc{Assg-L}] rule is
the one-sided rule for assignment. It can be derived from its two-sided
counterpart, the [\textsc{Assg}] rule, using the [\textsc{Struct}] rule.

The [\textsc{Rand-L}] rule is the one-sided rule for random sampling. It can be
derived from its two-sided counterpart, the [\textsc{Rand}] rule, using the
[\textsc{Struct}] rule and the fact that the product distribution
$\mu_1\times\mu_2$ is a valid coupling of proper distributions $\mu_1$ and
$\mu_2$.

We have one-sided and two-sided rules for conditionals.  The [\textsc{Cond-L}]
rule is the one-sided version; it can be derived from the [\textsc{Case}] and
[\textsc{Struct}] rules.  The [\textsc{Cond-S}] rule is the two-sided version;
the rule assumes that the two guards of the conditional statements are
synchronized, so that one must only need to reason about the cases where both
statements enter the true branch, and the case where both statements enter the
false branch.  It can be derived from the [\textsc{Case}], [\textsc{Struct}],
and [\textsc{False}] rules. 

We also have one-sided and two-sided rules for loops.  The
[\textsc{While-L}] rule corresponds to the \Sprhl\ one-sided rule for
\kwhile\ loops; it can be derived from the [\textsc{While}] rule by
setting $p_1=\top$, $p_0=p_2=\bot$, and $k_1=k_2=1$, $e=e_1$,
$e_2=\bot$, $c_2=\iskip$ and using the [\textsc{Struct}] and
[\textsc{False}] rules.  The rule [\textsc{While-S}] corresponds to
the \Sprhl\ two-sided rule for \kwhile\ loops. This rule assumes that
the two loops are synchronized, i.e., the guards of the two loops are
equal assuming the loop invariant. This rule can be derived from the
general [\textsc{While}] rule by setting $p_0=\top$, $p_1=p_2=\bot$,
and $k_1=k_2=1$, and $e=e_1$, and using the [\textsc{False}] rule.

\subsection{Soundness and Relative Completeness}
We have formally verified the soundness theorem below in the \coq
proof assistant, with its \ssreflect extension.
\begin{theorem}[Soundness]
Every derivable judgment is valid.
\end{theorem}

\begin{proof}[Proof sketch]
  The proof is by induction on the proof derivation.
  We only detail the case for the rule~\rname{While}, using the same
  notations of \cref{fig:rules}.
  It is immediate that the product program satisfies the
  post-condition:
  i)~the loop condition $e$ does not hold after the execution of the
  loop by construction, and
  ii)~the rule premises ensure that the loop body of the product
  program preserves the invariant $\post$.
  We are left to prove that the semantics of the projections of the
  product program correspond to respective semantics of the original
  programs. We here detail the proof for first projection, the one for
  the second projection being similar.

  Let $m$ s.t. $m \models \post$. We want to prove that
  $\sem{w_1}{m_1} = \proj_1(\sem{w}{m})$
  where $m_1 \triangleq \proj_1(m)$, $w \triangleq \iwhile{e}{c}$
  and $w_1 \triangleq \iwhile{e_1}{c_1}$.
  We prove this equation by verifying the double inequality
  $\sem{w_1}{m_1} \preceq \proj_1(\sem{w}{m})$ and
  $\proj_1(\sem{w}{m}) \preceq \sem{w_1}{m_1}$
  where $\preceq$ denotes the pointwise ordering of functions.
  We only detail the first one, the proof for the latter being
  similar.
  By definition of $\sem{w_1}{m_1}$ as
  $\lim_{n \to \infty} \sem{\itfor{c_1}{e_1,n}}{m_1}$,
  proving the first inequality can be reduced to proving that
  $\sem{\itfor{c_1}{e_1,n}}{m_1} \preceq \proj_1(\sem{w}{m})$
  for any $n \in \mathbb{N}$.
  We proceed by induction on $n$ and only detail the inductive case,
  the base case being immediate.
  We here do a case analysis on $\sem{e_1}{m}$. Here too, by lack of
  space, we only detail the more technical case where $\sem{e_1}{m}$
  holds.
  We then proceed by case analysis on $\sem{p_0}{m}$, $\sem{p_1}{m}$
  and $\sem{p_2}{m}$, all of them being pairwise mutually exclusive:
  \begin{itemize}
  \item If $\sem{p_0}{m}$ holds, then
    \begin{align*}
      \sem{\itfor{c_1}{e_1,n+1}}{m_1}
      &= \dlet {m'} {\sem{c_1}{m_1}} {\sem{\itfor{c_1}{e_1,n}}{m'}} \\
      &\preceq \dlet {m'} {\sem{\cfor{c_1}{e_1,k_1}}{m_1}}
                 {\sem{\itfor{c_1}{e_1,n}}{m'}} \\
      &= \dlet {m'} {\sem{c_0}{m}} {\sem{\itfor{c_1}{e_1,n}}{m'}}
       = \sem{w}{m}.
    \end{align*}

  \item If $\sem{p_1}{m}$ holds, then
    \[\left\{\begin{aligned}
      \sem{w}{m} &= \dlet {m'} {\sem{c_1}{m}} {\sem{w}{m'}} \\
      \sem{\itfor{c_1}{e_1,n+1}}{m_1} &=
        \dlet {m'} {\sem{c_1}{m_1}} {\sem{\itfor{c_1}{e_1,n}}{m'}}.
    \end{aligned}\right.\]
    We conclude by an immediate application of the induction
    hypothesis, using the monotony of the distribution expectation
    operator.

  \item If $\sem{p_2}{m}$ holds, then
    \[
      \sem{\itfor{c_1}{e_1,n+1}}{m_1}
      = \dlet{m'}{ \proj_1(\sem{s}{m})}{ \sem{\itfor{c_1}{e_1,n+1}}{m'}}
    \]
    where $s \triangleq \iwhile{e_2\land p_2}{c_2};w$.
    Note that after having executed $\iwhile{e_2\land p_2}{c_2}$
    ---which is lossless--- $\post$ and $e_1$ still hold while $e_2$
    does not. By the premises of the rule, $p_2$ must then be false,
    so $p_0 \oplus p_1$ holds for every memory $m'$ in the support of
    $\proj_1(\sem{s}{m})$.
    In that case, following the two first cases, we know that
    $\sem{\itfor{c_1}{e_1,n+1}}{m'} \preceq \proj_1(\sem{w}{m'})$.
    Hence,
    \begin{align*}
      \sem{\itfor{c_1}{e_1,n+1}}{m_1}
      &\preceq \dlet{m'}{\proj_1(\sem{s}{m})}{\proj_1(\sem{w}{m'})}
         \tag{monotony} \\
      &= \proj_1(\dlet{m'}{\sem{s}{m}}{\sem{w}{m'}})
         \tag{separability} \\
      &= \proj_1(\sem{w}{m}),
    \end{align*}
    where the penultimate step is valid because
    $\iwhile{e_2\land p_2}{c_2}$ does not modify $m_1$ ---~$c_1$ and
    $c_2$ being separable. \qedhere
  \end{itemize}
\end{proof}
Although it is not a primary objective of our work, we briefly comment
on completeness of the logic. First, the coupling method is not
complete for proving rapid mixing of Markov chains.\footnote{%
  \citet{DBLP:journals/rsa/KumarR01} show that the class of
  \emph{causal couplings}---which contains all couplings in our
  logic---are unable to prove rapid mixing for some rapidly-mixing Markov
  chains~\citep{diaconis91,JerrumS89}}
Second, it is not clear that our proof system is complete with respect to
hoisting random assignments out of loops.

However, we note that the \emph{deterministic} fragment of our logic achieves
completeness for programs that satisfy a sufficiently strong termination
property; the key is that the new rule for $\kwhile$ subsumes self-composition for
$\kwhile$ loops, provided they terminate on all initial memories satisfying the
invariant. More precisely, we can prove the following completeness theorem.

\begin{theorem}
  Let $c_1$ and $c_2$ be separable deterministic programs. If
  $\hoare{\pre}{c_1;c_2}{\post}$ is derivable using Hoare logic, then
  \[ \Equivw{c_1}{c_2}{\pre}{\post}{\_} \]
  is derivable. Therefore, \Sprprhl is relatively complete for
  deterministic programs.
\end{theorem}

\begin{proof}[Proof sketch]
  It suffices to prove that if $\hoare{\pre}{c_1}{\post}$ is derivable
  using Hoare logic, then
  $$\Equivw{c_1}{\iskip}{\pre}{\post}{\_}$$ is derivable. The proof
  proceeds by induction on the derivation.
\end{proof}

\subsection{Convergence from Couplings}
The fundamental theorem of couplings (\cref{thm:tv}) gives a powerful method to
prove convergence of random processes. First, we recast it in terms of \Sprprhl.
\begin{proposition}\label{prop:method}
Let $c_1$ and $c_2$ be separable programs and assume that the
following judgment is valid:
$$\Equivw{c_1}{c_2}{\pre}{\post\implies x_1=x_2}{c}$$
Then for every memory $m$ such that $m\models\pre$, we have
$$\TV(\mu_1,\mu_2)\leq \Pr_{x \sim \sem{c}{m}}[\neg \post]$$
where $\mu_1$, $\mu_2$ are the distributions obtained by sampling $x_1$ from
$\sem{c_1}{m}$ and $x_2$ from $\sem{c_2}{m}$ respectively.
\end{proposition}
This result is a direct consequence of the soundness of the logic, and allows
proving convergence in two parts.  First, we use \Sprprhl\ to establish a valid
post-condition of the form $\post\implies x_1=x_2$. Second, we prove that for
every memory $m$ satisfying some pre-condition, the product program $c$ built
from the derivation satisfies
$\Pr_{x \sim \sem{c}{m}}[\neg \post]\leq\beta$.
There are multiple approaches for proving properties of this form---reasoning
directly about the semantics of programs;
existing formalisms for bounding probabilities and reasoning about
expectations (e.g., \citep{Kozen85,Morgan96}); program logics for
probabilistic programs (e.g.,
\citep{Hartog:thesis,Ramshaw79,BGGHS16-icalp}). We will
check the property on pen and paper; mechanizing the proofs
is left for future work.

\section{Application: Convergence of Markov Chains}

We now turn to our first group of examples: proving convergence of probabilistic
processes.  Suppose we have a probabilistic process on a set $\Omega$ of
possible states. At each time step, the process selects the next state
probabilistically. Consider two runs of the same probabilistic process started
from two different states in $\Omega$. We would like to know how many steps we
need to run before the two distributions on states converges to a common
distribution. We consider several classic examples.

\paragraph*{Notation.} Throughout this section, we consider two
copies of the same program. To ensure that the two copies are
separable, we tag all the variables of the first copy with 1, and all
the variables of the second copy with 2. 

\subsection{Simple, Symmetric Random Walk}

\begin{wrapfigure}{left}{.18\textwidth}
\begin{algorithmic}
\State $\pos \gets \start$; $\imath\gets 0$; $l \gets []$;
\While {$\imath < T$}
  \State $\irnd{\jump}{\{ -1,1 \}}$; 
  \State $\pos \gets \pos + \jump$;
  \State $l \gets \jump::l$;
  \State $\imath \gets \imath+ 1$;
\EndWhile
\State \Return $\pos$
\end{algorithmic}
\caption{Random walk}\label{fig:rwalk}
\end{wrapfigure}

Our first example is a simple random walk on the integers. Let the
state space $\Omega$ be $\ZZ$. At each step, the process chooses
uniformly to move left (decreasing the position by $1$) or right
(increasing the position by $1$).  The program $\rwalk$ in
\Cref{fig:rwalk} implements the process, executed for $T$ steps.
The variable $l$ is a \emph{ghost variable}. While it does not
influence the process, it keeps track of the list of sampled values,
and will be used to state assertions when we construct the coupling.

Now, consider the random walk started from starting positions
$\start_1, \start_2$. We want to show that the $\TV$-distance
between the two distributions on positions decreases as we run for
more steps $T$; roughly, the random walk forgets its initial
position. It is not hard to see that if $\start_1 - \start_2$ is
an odd integer, then we will not have convergence: at any timestep
$t$, the support of one distribution will be on even integers while
the support of the other distribution will be on odd integers, so the
$\TV$-distance remains $1$.

When $\start_1 - \start_2 = 2k$ is even, we can construct a 
coupling to show convergence. \citet{BEGHSS15} used \Sprhl\ to couple these
random walks by \emph{mirroring}; informally, the coupled walks make mirror 
moves until they meet, when they make identical moves to stay equal. 
Specifically, they show that
\begin{align*}
  \Equiv{\rwalk_1}{\rwalk_2}{\pre}{
    k\in \psum (\rev(l_1))\implies \pos_1 = \pos_2}
\end{align*}
where $\pre\triangleq \start_1 - \start_2 = 2k$, $\rev(l)$ reverses the list,
and $\psum(l)$ is the list of partial sums of $l$ (sums over its initial
segments). We can lift the judgment to \Sprprhl
\begin{align*}
    \Equivw{\rwalk_1}{\rwalk_2}
    {\pre}
    {k\in \psum (\rev(l_1))\implies \pos_1 = \pos_2}
    {\rwalk_0}
  \end{align*}
  where $\rwalk_0$ is the following product program:
  \begin{algorithmic}
    \State $\pos_1 \gets \start_1$;~ $\pos_2 \gets \start_2$;~
     $\imath_1 \gets 0$;~ $\imath_2 \gets 0$;~
     $l_1 \gets []$;~ $l_2 \gets []$;
    \While {$\imath_1 < T$}
    \If {$\pos_1 = \pos_2$}
    \State $\irnd{(\jump_1,\jump_2)}{\CpId{\{-1,1\}}}$;
    \Else\ \State $\irnd{(\jump_1,\jump_2)}{\CpDist{\opp}{\{-1,1\}}}$; 
    \EndIf
    \State $\iass{\pos_1}{\pos_1 + \jump_1}$;~
           $\iass{\pos_2}{\pos_2 + \jump_2}$;
    \State $\iass{l_1}{\jump::l_1}$;~
           $\iass{l_2}{\jump::l_2}$; 
    \State $\imath_1\gets \imath_1 + 1$;~
           $\imath_2\gets \imath_2 + 1$;
    \EndWhile
    \State \Return $(\pos_1, \pos_2)$
  \end{algorithmic}
  where $\opp x \triangleq -x$.
  
We briefly sketch the derivation of the product.  We start by an
application of the [\textsc{While-S}] rule with the invariant
\begin{align*}
  \Psi &\triangleq k \in \psum (\rev(l_1))\implies \pos_1 = \pos_2 \\
    &\land k \notin \psum (\rev(l_1)) \implies \pos_1 - \pos_2 = 2k - \ssum(l_1) ,
\end{align*}
where $\ssum(l)$ is the sum of the list.  Then, we apply a [\textsc{Seq}] rule
(consuming the first random sampling on each side) with intermediate assertion
$$\Xi \triangleq \post \land (\pos_1 = \pos_2 \implies \jump_1=\jump_2) \land
(\pos_1 \neq \pos_2 \implies \jump_1= -\jump_2)$$
The sub-proof obligation on tails is straightforward, the interesting
one is for the random sampling. We start by using the [\textsc{Case}]
rule with $e \triangleq \pos_1 = \pos_2$, which introduce the
conditional in the product program. If the equality holds, then the
two random values are synchronized using the $\CpId{\{-1,1\}}$; if
not, they are mirrored using $\CpDist{\opp}{\{-1,1\}}$. In our mirror
coupling, $k\in \mathsf{psum}(l_1)$ implies that the walks have
already met, and continue to have the same position.

However, the derivation by itself does not tell us \emph{how far} the
two distributions are, as a function of $T$. To get this information, we will
use the probabilistic product construction and the following classical result
from the theory of random walks.

\begin{theorem}[e.g., \citep{randomwalk:notes}] \label{thm:rw}
  Let $X_0, X_1, \dots$ be a symmetric random walk on the integers with initial
  position $X_0 = 0$. Then, for any position $k \in \mathbb{Z}$, the probability
  that the walk does not reach $k$ within $t$ steps is at most
  \[
    \Pr [ X_0, \dots, X_t \neq k ] \leq \frac{k e \sqrt{2}}{\pi \sqrt{t}} .
  \]
\end{theorem}

Now, we can analyze how quickly the two walks mix.

\begin{theorem} \label{thm:rw-converge}
  If we perform a simple random walk for $T$ steps from two positions
  that are $2k$ apart, then the resulting distributions $\mu_1$ and
  $\mu_2$ on final positions satisfy $\TV(\mu_1, \mu_2) \leq \frac{k e
    \sqrt{2} }{\pi \sqrt{T}}$.  Formally, for every two memories $m_1$
  and $m_2$ such that $m_1(\pos)-m_2(\pos)= 2 k$, we have
  \[
  \TV(\sem{\rwalk}{m_1},\sem{\rwalk}{m_2}) 
  \leq
  \frac{k e \sqrt{2} }{\pi \sqrt{T}},
  \]
\end{theorem}
\begin{proof}
  Conceptually, we can think of the difference $\pos_1 - \pos_2$ as a random
  walk which increases by $2$ with half probability and decreases by $2$ with
  half probability.  By applying \cref{thm:rw} to this random walk, we find that
  in the product program $c_0$,
  \[
    \Pr [ k \notin \psum (l_1) ] \leq \frac{k e \sqrt{2}}{\pi \sqrt{T}} .
  \]
  Then, we can conclude by \cref{prop:method}.
\end{proof}

\subsection{The Dynkin Process}
Our second example models a process called the \emph{Dynkin process}.
There is a sequence of $N$ concealed cards, each with a number drawn
uniformly at random from $\{ 1, \dots, 10 \}$. A player starts at some
position in $\{ 1, \dots, 10 \}$.  Repeatedly, the player looks at the
number at their current position, and moves forward that many steps.
For instance, if the player reveals $2$ at their current location,
then she moves forward two spaces. The player stops when she passes
the last card of the sequence. We want to show fast convergence of
this process if we start from any two initial positions. In code, the
Dynkin process is captured by the program $\dynkin$ defined in
\Cref{fig:dynkin}. Here, $l$ stores the history of positions of the
player; this ghost variable will be useful both for writing assertions
about the coupling, and for assertions in the product program.  Just
as for random walks, we can consider the mixing rate of this
process, starting from two positions $\start_1, \start_2$.
\begin{wrapfigure}[8]{left}{.20\textwidth}
\begin{algorithmic}
\State $\pos \gets \start$; $l \gets [\pos]$;
\While {$\pos < N$}
\State $\irnd {\jump} {\intv{1}{10}}$;
\State $\pos \gets \pos + \jump$;
\State $l \gets \pos::l$;
\EndWhile
\State \Return $\pos$
\end{algorithmic}
\caption{Dynkin process}\label{fig:dynkin}
\end{wrapfigure}
We will couple two runs of the Dynkin process from two starting
positions asynchronously: we will move whichever
process is behind, holding the other process temporarily fixed.  If
both processes are at the same position, then they move together.

\vspace*{1.8em}

Formally, we 
prove the following \Sprprhl judgment:
$$
  \Equivw{\dynkin_1}{\dynkin_2}{\top}{(l_1 \cap l_2 \neq \emptyset) \implies
    \pos_1=\pos_2}{\dynkin_0}
$$
where $\dynkin_0$ is the following program:
  \begin{center}
    \begin{algorithmic}
      \State $\pos_1 \gets \start_1$;~ $\pos_2 \gets \start_2$;
      \State $l_1 \gets {[\pos_1]}$;~ $l_2 \gets {[\pos_2]}$;
      \While {$\pos_1 < N \lor \pos_2 < N$}
      \If {$\pos_1 = \pos_2$}
      \State $\irnd{(\jump_1,\jump_2)}{\CpId{\intv{1}{10}}}$; 
      \State $\pos_1 \gets \pos_1 + \jump_1$;~
             $\pos_2 \gets \pos_2 + \jump_2$;
      \State $l_1 \gets \pos_1 :: l_1$;~  $l_2 \gets \pos_2 :: l_2$; 
      \Else \If {$\pos_1 < \pos_2$}
      \State $\irnd{\jump_1}{\intv{1}{10}}$;
      \State $\pos_1 \gets \pos_1 + \jump_1$;
      \State $l_1 \gets \pos_1::l_1$;
      \Else
      \State $\irnd{\jump_2}{\intv{1}{10}}$;
      \State $\pos_2 \gets \pos_2 + \jump_2$;
      \State $l_2 \gets \pos_2::l_2$;
      \EndIf
      \EndIf
      \EndWhile
      \State \Return $(\pos_1, \pos_2)$
    \end{algorithmic}
  \end{center}
 To couple the two programs, we use the rule  
 [\textsc{While}] with $k_1 = k_2 = 1$. To control which process
 will advance, we define:
  \begin{mathpar}
    p_0 \triangleq (\pos_1 = \pos_2)
    \and
    p_1 \triangleq (\pos_1 < \pos_2)
    \and
    p_2 \triangleq (\pos_1 > \pos_2) .
  \end{mathpar}
The lossless conditions are satisfied since the distance between
$x_1$ and $x_2$ strictly decreases at each iteration of the loops.
For the samplings, when $p_1$ or $p_2$ hold we use the one-sided rules for
random sampling ([\textsc{Rand-L}] and corresponding [\textsc{Rand-R}]);
when $p_0$ holds, we use the identity coupling.

\begin{theorem} \label{thm:dynkin-converge}
  Let $m_1$ and $m_2$ two memories such that \\
  $m_1(\start), m_2(\start)\in \intv{1}{10}$, and suppose $N > 10$. Then:
  \[
    \TV(\sem{\dynkin}{m_1},\sem{\dynkin}{m_2}) \leq (\nicefrac{9}{10})^{\nicefrac{N}{5} - 2} .
  \]
  The two distributions converge exponentially fast as $N$ grows.
\end{theorem}
\begin{proof}
  In the product program $\dynkin_0$, we want to bound the probability
  that $l_1$ and $l_2$ are disjoint; i.e., the probability that the
  two processes never meet. We proceed in two steps. First, it is not
  hard to show that $|\pos_1 - \pos_2| < 10$ is an invariant.  Thus,
  at each iteration, there is a $\nicefrac{1}{10}$ chance that the
  lagging process hits the leading process. Second, each process moves
  at most $10$ positions each iteration and we finish when both
  processes reach the end, so there are at least $2(\nicefrac{N}{10} -
  1) = \nicefrac{N}{5} - 2$ chances to hit.  Therefore, in the product
  we can show
  \begin{equation} \label{eq:dynkin-not-meet}
    \Pr[ (l_1 \cap l_2) = \emptyset ] \leq (\nicefrac{9}{10})^{\nicefrac{N}{5} - 2} .
  \end{equation}
  We can then conclude by \cref{prop:method}.
\end{proof}

To highlight the quantitative information verified by our approach, we note that
the corresponding theorem for random walks (\cref{thm:rw-converge}) shows that
the total variation distance between two random walks decreases as $O(1/\sqrt{T})$. In contrast,
\cref{thm:dynkin-converge} shows that total variation distance between two
Dynkin processes converges as $O(0.9^N)$, giving much faster mixing
(exponentially fast instead of polynomially fast).

\begin{remark*}
  The Dynkin process is inspired by the following two-player game. Both players
  pick a starting position. There is one sequence of random cards that is shared
  by both players, and players look at the card at their current position and
  move forward that many spaces. The random cards are shared, so a player
  samples the card only if the other player has not yet visited the position. If
  the other player has already landed on the position, the later player observes
  the revealed card and moves forward.

  While the Dynkin process samples every card that the player lands on, it is
  not hard to see that the product program exactly models the two-player game.
  More specifically, the product interleaves the players so that at each turn,
  the player that is lagging behind makes the next move. By this scheduling, as
  long as one player is lagging behind, the players have not landed on the same
  position and so each player lands on unseen cards and draws random samples to
  reveal. Once the players meet, the product program makes the same move for
  both players.

  Under this interpretation, \cref{eq:dynkin-not-meet} bounds the probability
  that the players do not land on the same final position from any two initial
  positions.  This result is the basis of the magic trick called \emph{Dynkin's
    card trick} or the \emph{Kruskal count}. If one player is the magician and
  the other player is a spectator, if the spectator starts at a secret position
  and runs the process mentally, the magician can guess the correct final
  position with high probability by starting at any position and counting along.
\end{remark*}

\section{Application: Path Coupling}

So far, we have seen how to prove convergence of probabilistic processes by
constructing a coupling, reasoning about the probability of the processes not
meeting, and applying the coupling theorem. While this reasoning is quite
powerful, for more complicated processes it may be difficult to directly construct a
coupling that shows fast mixing. For example, it can be difficult to find and
reason about a coupling on the distributions from two states $s, s'$ if there
are many transitions apart in the Markov chain.

To address this problem, \citet{bubley1997path} proposed the \emph{path
  coupling} technique, which allows us to consider just pairs of \emph{adjacent}
states, that is states where $s$ can transition in one step to $s'$. Roughly
speaking, if we can give a good coupling on the distributions from two adjacent
states for every pair of adjacent states, then path coupling shows that the
state distributions started from two arbitrary states converge quickly.

As the name suggests, path coupling considers paths of states in a probabilistic
process. For this to make sense, we need to equip the state space with
additional structure.  For the basic setup, let $\OState$ be a finite set of
states and suppose that we have a metric $d : \OState \times \OState \to \NN$.
We require that $d$ is a \emph{path metric}: if $d(s, s') > 1$, then there
exists $s'' \neq s, s'$ such that $d(s, s') = d(s, s'') + d(s'', s')$. Two
states are said to be \emph{adjacent} if $d(s, s') = 1$. We will assume that the
\emph{diameter} of the state space, i.e. the distance between any two states, is
finite: $\Delta < \infty$. The Markov chain is then defined by iterating a
transition function $f : \OState \to \Dist(\OState)$.

The main idea behind path couplings is that if we can couple the
distributions from any two adjacent states, then there exists a
coupling for the distributions from two states at an arbitrary
distance, constructed by piecing together the couplings between them.
Furthermore, if the expected distance between states contracts
under the coupling on adjacent states, i.e., the resulting expected
distance is strictly less than $1$, then the same holds for the
coupling on two states at any distance. More formally, we have the
following main theorem of path coupling. 
\begin{lemma}[\citet{bubley1997path}]
  \label{lem:pathcoupling}
  Consider a Markov chain with transition function $f$ over a set
  $\OState$ with diameter at most $\Delta$.
  Suppose that for any two states $s$ and $s'$ such that $d(s, s') =
  1$, there exists a coupling $\mu$ of $f(s), f(s')$ such that
  $\E {(r, r')} {\mu} {d(r, r')} \leq \beta$.

  Then, starting from any two states $s$ and $s'$ and running $t$ steps of the
  chain, there is a coupling $\mu$ of $f^t(s), f^t(s')$ such that
  \begin{align*}
    \TV(f^t(s), f^t(s')) \leq \P {(r, r')} {\mu} {r \neq r'} \leq \beta^t \Delta.
  \end{align*}
\end{lemma}
This lemma applies for all $\beta$, but is most interesting for $\beta
< 1$ when it implies that the Markov chain mixes quickly. With the
main theorem in hand, we will show how to verify the conditions for
path coupling on two examples from \citet{bubley1997path}.

\subsection{Graph Coloring: the Glauber Dynamics}

Our first example is a Markov chain to provide approximately uniform
samples from the set of colorings of a finite graph; it was first analyzed by
\citet{DBLP:journals/rsa/Jerrum95}, our analysis follows \citet{bubley1997path}.
Recall that a
finite \emph{graph} $G$ is defined by a finite set of vertices $V$,
and a symmetric relation $E$ relating pairs vertices that are
connected by an edge; we will let $\neighbors{G}(v) \subseteq V$
denote the set of neighbors of $v$, i.e the set vertices that have an
edge to $v$ in $G$.  Let $C$ be the set of \emph{colors}; throughout,
we assume that $C$ is finite.  A \emph{coloring} $w$ of $G$ is a map
from $V$ to $C$. A coloring is \emph{valid} (sometimes called
\emph{proper}) if all neighboring vertices have different colors: for
all $v' \in \neighbors{G}(v)$ we have $w(v) \neq w(v')$. We write
$\valid{G}(w)$ if $w$ is a valid coloring.
The following program $\glauber(T)$ models $T$ steps of the
\emph{Glauber dynamics} in statistical physics:
\begin{center}
  \begin{algorithmic}
$\imath\gets 0$;
\While{$\imath<T$}
\State $\irnd{v}{V}$; $\irnd{c}{C}$;
\State \SIfT {$\valid{G}(w[v \mapsto c])$}{$w \gets w[v \mapsto c]$};
\State $\imath\gets \imath+1$
\EndWhile
\end{algorithmic}
\end{center}
Informally, the algorithm starts from a valid coloring $w$
and iteratively modifies it by sampling uniformly a vertex $v$
and a color $c$, recoloring the $v$ with $c$ if this continues to be a valid
coloring.

We want to measure the statistical distance between two executions of the
process starting from two initial colorings $w_1$ and $w_2$. There are two
natural approaches. The first option is to couple the two copies of $\glauber$
directly, analyze the product program and apply \cref{prop:method}. The problem
is that when the two colorings are far apart, it is hard to reason about how the
processes might meet under a coupling; \citet{DBLP:journals/rsa/Jerrum95}
adopted this strategy, but the resulting proof is dense and complex.

The second, far simpler option is to apply path coupling. Here, we build a product for just
one iteration of the loop, and it suffices to consider cases where the two
initial states are adjacent. This drastically simplifies the coupling and
analysis of the product program, so we adopt this approach here.  For the sake
of clarity, we adapt the transition function so that its output is stored in a
fresh variable $w'$, and call the resulting program $\glauber^\dagger$:
\begin{center}
\begin{algorithmic}
\State $\irnd{v}{V}$; $\irnd{c}{C}$;
\State \SIf {$\valid{G}(w[v \mapsto c])$}{$w' \gets w[v \mapsto c]$}{$w' \gets w$}
\end{algorithmic}
\end{center}
Note that $\glauber^\dagger; w\gets w'; \imath\gets \imath +1$ is semantically
equivalent to the loop body of $\glauber$.

To apply the path coupling theorem \cref{lem:pathcoupling}, we need to
define a path metric on $\Omega$ and construct a coupling for the
process started from two adjacent states. For the path metric, we
define the distance $d(w_1, w_2)$ to be the Hamming distance: the
number of vertices where $w_1$ and $w_2$ provide different colors. We
say two states are \emph{adjacent} if $d(w_1, w_2) = 1$; these
states differ in the color of exactly one vertex.  In order to apply
path coupling, we need to find a coupling of the transition function
on adjacent states such that the expected distance shrinks. We first
build the coupling using \Sprprhl. Specifically, we prove
\begin{equation} \label{eq:glbr-judge}
  \Equivw{\glauber_1^\dagger}{\glauber_2^\dagger}{d(w_1, w_2) = 1}{d(w_1', w_2') \leq 2}{c_0}
\end{equation}
where $c_0$ is the following program:
\begin{center}
\begin{algorithmic}
\State $\irnd{v_1,v_2}{\CpId{V}}$;
\State \SIf{$v_1 \in \neighbors{G}(v_0)$}{$\irnd{c_1,c_2}{\CpDist{\pi^{\ab}}{C}}$}
                 {$\irnd{c_1,c_2}{\CpId{C}}$};
\State \SIf {$\valid{G}(w_1[v_1 \mapsto c_1])$}
            {$w'_1 \gets w_1[v_1 \mapsto c_1]$}
            {$w'_1 \gets w_1$};
\State \SIf {$\valid{G}(w_2[v_2 \mapsto c_2])$}
            {$w'_2 \gets w_2[v_2 \mapsto c_2]$}
            {$w'_2 \gets w_2$}
\end{algorithmic}
\end{center}

We briefly sketch how to prove the judgment. Note that the
two states must agree at all vertices, except at a single vertex
$v_0$. Let $w_1(v_0) = a$ and $w_2(v_0) = b$. First, we couple the
vertex sampling with the rule [\textsc{Rand}] using the identity
coupling, ensuring $v_1 = v_2$. Then, we use the rule [\textsc{Case}]
to perform a case analysis on the sampled vertex, call it $v$.  If $v$
is a neighbor of the differing vertex $v_0$, we use the rule
[\textsc{Rand}] and the transposition bijection $\pi^{\ab} : C \to C$
defined by the clauses:
  \begin{mathpar}
    \pi^{\ab}(a) = b
    \and
    \pi^{\ab}(b) = a
    \and
    \pi^{\ab}(x) = x \quad \text{otherwise}
  \end{mathpar}
to ensure that $c_1 = \pi^{\ab}(c_2)$.  Otherwise, we use the rule
[\textsc{Rand}] and the identity coupling to ensure $c_1 = c_2$. By
applying the one-sided rules for conditionals to the left and the
right sides ([\textsc{Cond-L}] and [\textsc{Cond-R}]), we can
conclude the derivation.
  
Next, we must reason about the expected value of the distance between
$w'_1$ and $w'_2$ after executing the product program.
\begin{lemma}
  Let $n = |V|$ and $k = |C|$, and suppose that the graph $G$ has
  degree bounded by $D$.  That is, for any $v \in V$, there are at
  most $D$ vertices $v'$ such that $E(v, v')$. If $k \geq D$ , then
  there is a coupling $\mu$ of the distributions after running
  $\glauber^\dagger$ on adjacent states such that
  $$
    \sE {\mu} {d(w_1', w_2')}  \leq 1 - 1/n + 2D/kn.
  $$
\end{lemma}
\begin{proof}
  Let $w_1$ and $w_2$ be adjacent states. We must bound the expected distance
  between the states $w_1', w_2'$ in the product program. Let $d' = d(w_1',
  w_2')$, we have:
  \begin{align*}
    & \sE {\mu} {d'} =
        0 \cdot \Pr[d'=0] + 1 \cdot \Pr[d'=1] + 2 \cdot \Pr[d'=2] \\
    &\quad = 1 - \Pr [  d' = 0 ] + \Pr [  d' = 2 ] \\
    &\quad \leq 1 - \Pr [ v_1 = v_0 \land \valid{G}(w' )] 
            + \Pr [ v_1 \in \neighbors{G}(v_0) \land c_1 = b ] \\
    &\quad \leq 1 - \frac{1}{n} \left( 1 - \frac{D}{k} \right) + \frac{D}{nk}
       = 1 - \frac{1}{n} + \frac{2 D}{nk} .
  \end{align*}
  where $w' = w_1[v_0 \mapsto c_1]$.
  The first equality holds because the distance between the two 
  resulting coloring will be at most $2$ by judgment \cref{eq:glbr-judge}.
  The second equality holds since $1 = \Pr[d'=0] + \Pr[d'=1] + \Pr[d'=2]$.
  The second to last step follows since each vertex has at most $D$ neighbors,
  so there are at least $k - D$ valid colors at any vertex.
\end{proof}

Applying the path coupling lemma (\cref{lem:pathcoupling}), noting that the
diameter is $n$ since there are $n$ vertices, proves that the Glauber
dynamics mixes quickly if there are sufficiently many colors $k$.

\begin{theorem}
  Consider the Glauber dynamics on $k$ colors with a graph $G$ with $n$ vertices
  and degree at most $D$, and suppose $k \geq 2 D + 1$. Then, for some
  constant $\beta < 1$,
  \[
    \TV( \sem{\glauber(T)}{m_1}, \sem{\glauber(T)}{m_2} ) \leq \beta^T n
  \]
  for any two initial memories $m_1$ and $m_2$ containing valid colorings.
\end{theorem}
This theorem recovers the result by \citet{bubley1997path}; this is the key
step to showing that running Glauber dynamics for a small number of steps and
taking a sample is almost equivalent to drawing a uniformly random sample from
all proper colorings of the graph.

\subsection{Independent Sets: the Conserved Hard-Core Model}

Our second example is from graph theory and statistical physics, modeling the
evolution of a physical system in the \emph{conserved hard-core lattice gas}
(CHLG) model~\citep{bubley1997path}. Suppose that we have a set $P$ of particles, and we have a graph
$G$. A \emph{placement} is a map $w : P \to V$, placing each particle at a
vertex of the graph. We wish to place the particles so that each vertex has at
most one particle, and no two particles are located at adjacent vertices; we
call such a placement \emph{safe} and denote it by $\safe{G}(w)$. For a specific
graph, there could be multiple safe placements.

If we want to sample a uniformly random safe placement, we can use a simple
Markov chain. We take the state space $\Omega = P \to V$ to be the set of
placements. Again, we take $G$ and $P$ to be finite. 
We start using a safe initial placement. Each step, we sample a
particle $p$ from $P$ and a vertex $v$ from $V$ uniformly at random and try to
place $p$ at $v$. If $w[p \mapsto v]$ is safe, then we make this the new
placement; otherwise, we keep the same placement. We can model $T$ steps of this
dynamics with the following
program $\chlg(T)$:
\begin{center}
  \begin{algorithmic}
$\imath\gets 0$;
\While{$\imath<T$};
\State $\irnd{p}{P}$; $\irnd{v}{V}$;
\State \SIfT {$\safe{G}(w[p \mapsto v])$}{$w \gets w[p \mapsto v]$}
\State $\imath\gets \imath+1$
\EndWhile
\end{algorithmic}
\end{center}
As in the previous example, we adapt the loop body to form $\chlg^\dagger$:
\begin{center}
\begin{algorithmic}
\State $\irnd{p}{P}$; $\irnd{v}{V}$;
\State \SIf {$\safe{G}(w[p \mapsto v])$}
            {$w' \gets w[p \mapsto v]$}
            {$w' \gets w$}
\end{algorithmic}
\end{center}

Like the graph coloring sampler, we take the path metric on placements to be
Hamming distance and try to find a coupling on the distributions from adjacent
initial placements.

\begin{lemma}
  Let $s = |P|$ and $n = |V|$, and suppose that the graph $G$ has degree bounded
  by $D$. Starting from any two adjacent safe placements $w_1$ and $w_2$, there
  is a coupling $\mu$ on the distributions after one step such that
  \begin{align*}
    \E {(w_1, w_2)} \mu {d(w_1', w_2')}  \leq \left( 1 - \frac{1}{s} \right)
    \left( \frac{ 3 (D + 1) }{ n } \right) .
  \end{align*}
\end{lemma}
\begin{proof}
  Let $\chlg^\dagger_1, \chlg^\dagger_2$ be two copies of the transition
  function, with variables tagged.  Consider two adjacent placements $w_1$ and
  $w_2$. We will sketch how to couple the transitions.
  
  We use rule [\textsc{Rand}] twice to couple the particle and vertex samplings
  with the identity coupling, ensuring $p_1 = p_2$ and $v_1 =
  v_2$. Then, we can apply the one-sided rules for conditionals to the left
  and the right sides ([\textsc{Cond-L}] and [\textsc{Cond-R}]) to conclude the
  following judgment:
  \begin{align*}
    \Equivw{\chlg^\dagger_1}{\chlg^\dagger_2}{d(w_1, w_2) = 1}{d(w_1, w_2) \leq 2}{c_0} ,
  \end{align*}
  where $c_0$ is the following product program:
  \begin{center}
    \begin{algorithmic}
      \State $\irnd{(p_1, p_2)}{\CpId{P}}$;
      $\irnd{(v_1, v_2)}{\CpId{V}}$;
      \State
      \SIf {$\safe{G}(w_1[p_1 \mapsto v_1])$}
      {$w_1' \gets w_1[p_1 \mapsto v_1]$}
      {$w_1' \gets w_1$};
      \State
      \SIf {$\safe{G}(w_2[p_2 \mapsto v_2])$}
      {$w_2' \gets w_2[p_2 \mapsto v_2]$}
      {$w_2' \gets w_2$}
    \end{algorithmic}
  \end{center}
  Then, we can bound the expected distance between $w_1', w_2'$ in the product
  program. Let $d' = d(w'_1,w'_2)$, we have:
  \begin{align*}
    & \sE {(w_1, w_2)} {\mu} {d'} \\
    &\; = 1 - \Pr [d' = 0] + \Pr [d' = 2] \\
    &\; = 1 - \Pr [ p = p_0 \land \safe{G}(w_1[p \mapsto v])] \\
    &\; \quad + \Pr [ p \neq p_0 \land \neg(\safe{G}(w_1[p \mapsto v])
                   \iff \safe{G}(w_2[p \mapsto v]))] \\
    &\; \leq 1 - \Pr [ p = p_0 \land \safe{G}(w_1[p \mapsto v])] \\
    &\; \quad + \Pr [ p \neq p_0 \land
          \neg(\safe{G}(w_1[p \mapsto v]) \land \safe{G}(w_2[p \mapsto v]))].
  \end{align*}
  We can bound the two probability terms. For the first term, we know that the
  probability of selecting $p = p_0$ is $1/s$, and the probability that $p$ is
  safe at $v$ if it avoids all other points (at most $s - 1$) and all the
  neighbors of the other points (at most $(s - 1) D$); this probability is the
  same for both placements $w_1$ and $w_2$, since the two placements are
  identical on points besides $p_0$.

  For the second term, we know that the probability of selecting $p \neq p_0$ is
  $1 - 1/s$, and $p$ is not safe at $v$ in placement $w_1$ or in $w_2$ if we
  select the position $a$, $b$, or one of their neighbors. Putting everything
  together, we can conclude:

  \begin{align*}
    &\sE {\mu} {d(w_1', w_2')} \\
    &\quad \leq 1 - \frac{1}{s} \left( 1 - \frac{(s - 1)(D + 1)}{n} \right) \\
    &\quad \quad + (s - 1) \left(
        \frac{ |\neighbors{G}(a)| + 1 + |\neighbors{G}(b)| + 1 }{ sn }
      \right) \\
    &\quad \leq 1 - \frac{1}{s} \left(
        1 - \frac{(s - 1)(D + 1)}{n}
      \right) + (s - 1) \left( \frac{ 2(D + 1) }{ sn } \right) \\
    &\quad = \left(1 - \frac{1}{s} \right) \left( \frac{3 (D + 1) }{n} \right) .
    \qedhere
  \end{align*}
\end{proof}

Applying the path coupling lemma (\cref{lem:pathcoupling}) shows that if we
iterate the transition function on two initial placements, the resulting
distributions on placements converge quickly.

\begin{theorem}
  Consider the conserved lattice gas model with $s = |P|$ particles on a graph
  $G$ with $n = |V|$ vertices and degree at most $D$, where $s \leq n/3(D + 1) +
  1$. Then, for a constant $\beta < 1$,
  \[
    \TV( \sem{\chlg}{m_1}, \sem{\chlg}{m_2}) \leq \beta^T s
  \]
  for any two initial memories $m_1, m_2$ containing safe placements.
\end{theorem}

\begin{remark*}
This theorem is slightly weaker than the corresponding result by
\citet{bubley1997path}, who prove rapid mixing under the weaker condition $s
\leq n/2(D + 1) + 1$. Roughly, they use the maximal coupling on the two
transition distributions, giving a tighter analysis and better bound. It is also
possible use the maximal coupling in \Sprprhl, but the corresponding
specification of the coupling distribution would be proved as part of soundness
of the logic, rather than as a property of a probabilistic program.
\end{remark*}

\section{Application: Loop Optimizations}
Program equivalence is one of the original motivation for relational program
logics~\citep{Benton04}. In this section, we demonstrate the effectiveness of our logic
using several examples of exact and approximate program equivalence.
Our first example is a loop transformation which originates from the
literature on parallelizing compilers but also has applications in
computer-aided cryptography. Our second example is drawn from the
recent literature on approximate computing, and is an instance of loop
perforation.

\subsection{Loop Strip-Mining}
Loop \emph{strip-mining} (or \emph{loop sectioning}) is a transformation that turns
a loop into a nested loop. While the transformation originates from
literature on parallelizing compilers and is primarily used to take
advantage of vectorized instructions, it is also useful for formally proving
the computational security of certain cryptographic constructions.
The following example, \cref{fig:loop}, is inspired from a proof of 
indistinguishability of the SHA3 hash function~\citep{CanteautFNPRV12}.\footnote{%
  for simplicity, the programs use an operator $f$ which takes randomness as an
  argument (note that the value $r$ is sampled immediately before the assignment
  using $f$), although in the proof of the SHA3 hash function $f$ is a procedure
  call whose body performs random samplings.}
Using the rule for $\kwhile$ loops, we can prove the following:
$$\Equivw{c_1}{c_2}{x_1=x_2}{x_1=x_2}{c}$$
The crux of the proof is applying the [\textsc{While}] rule with
$k_1=1$ and $k_2=M$, and $e=\imath_1 < N$, and $p_0 = \top$, 
and $p_1 = p_2 = \bot$ and an invariant $\post$
which strengthens the assertion $x_1=x_2$ mainly by adding 
$l_2 = \imath_1 \cdot M$.
Side conditions using $p_1$ and $p_2$ are trivial to prove (using the
\textsc{False} rule) since they have $\bot$ in hypothesis. It remains to check
the premise for for $p_0$, but we now have two synchronized loops; we can use
the [\textsc{Struct}] rule to remove the conditional on $p_0$ which is always
true in this case.

\begin{figure}
  \begin{minipage}[t]{0.49\linewidth}
  Nested loop:
  {{\begin{algorithmic}
  \State $\imath_1 \gets 0$;
  \While{$\imath_1 < N$}
  \State $\jmath_1 \gets 0$;
  \While {$\jmath_1 < M$}
  \State $l_1 \gets \imath_1 \cdot M + \jmath$;
  \State $\irnd{r_1}{\mu}$;
  \State $x_1 \gets f(l_1,x_1,r_1)$;
  \State $\jmath_1 \gets \jmath_1 + 1$;
  \EndWhile
  \State $\imath_1 \gets \imath_1  + 1$;
  \EndWhile
  \State $l_1 \gets N \cdot M$;
  \end{algorithmic}
      }}
  \end{minipage}
  \begin{minipage}[t]{0.49\linewidth}
  Single loop:
  {{\begin{algorithmic}
  \State $l_2 \gets 0$;
  \While {$l_2 < N \cdot M$}
  \State $\imath_2 \gets l_2 \mathop{/} M $;
  \State $\jmath_2 \gets l_2 \mathop{\%} M$;
  \State $\irnd{r_2}{\mu}$;
  \State $x_2 \gets f(l_2,x_2,r_2)$;
  \State $l_2 \gets l_2 + 1$;
  \EndWhile
  \State $\imath_2 \gets N$;
  \State $\jmath_2 \gets M$;
    \end{algorithmic}
  }}
  \end{minipage}
\\[2ex]
  Product program:
  {{
      \begin{algorithmic}
  \State $\imath_1 \gets 0$; $l_2 \gets 0$;
  \While {$\imath_1 < N$}
  \State $\jmath_1 \gets 0$;
  \While {$\jmath_1 < M$}
  \State $l_1 \gets \imath_1 \cdot M + \jmath_1$;
  \State $\imath_2 \gets l_2 \mathop{/} M $;~
         $\jmath_2 \gets l_2 \mathop{\%} M$;
  \State $\irnd{(r_1, r_2)}{\CpId{\mu}}$;
  \State $x_1 \gets f(l_1,x_1,r_1)$; $x_2 \gets f(l_2,x_2,r_2)$;
  \State $\jmath_1 \gets \jmath_1 + 1$;~
         $l_2 \gets l_2 + 1$;
  \EndWhile
  \State $\imath_1 \gets \imath_1 + 1$;
  \EndWhile
  \State $l_1 \gets N \cdot M$;~
         $\imath_2 \gets N$;~
         $\jmath_2 \gets M$;
  \end{algorithmic}
  }}
\caption{Loop strip-mining}\label{fig:loop}
\end{figure}

\subsection{Loop Perforation}
Loop perforation~\citep{SidiroglouMHR11,MisailovicRR11} is a program
transformation that delivers good trade-offs between performance and
accuracy, and is practical in many applications, including image and
audio processing, simulations and machine-learning. Informally, loop
perforation transforms a loop that performs $n$ iterations of its body
into a loop that performs $m<n$ iterations of its body, followed by a
simple post-processing statement. \Cref{fig:loopperf} shows an example
of loop perforation inspired from a financial analysis application,
called swaptions. In this example, every other loop iteration is
skipped, and the post-processing statement simply multiplies by 2 the
value $s$ computed by the optimized loop. As for the previous example,
we can prove the following judgment:
$$\Equivw{c_1}{c_2}{\top}{s_1=s_2}{c}$$
The product program can be built using the [\textsc{While}] rule.  We use $e =
\imath_1 < 2 \cdot n$, $p_0=\top$, $p_1 = p_2 = \bot$ and $k_1 = 2$ and $k_2 = 1$, and
the invariant is $\imath_1 = 2 \cdot \imath_2$.  The invariant allows to show, using
the [\textsc{Struct}] rule, that the loop on $k_1$ (denoted by $\cfor{c_1}{e_1,
  k_1}$ in [\textsc{While}]) perform exactly $2$ iterations.  Using the product,
one can also analyze the (probabilistic) accuracy rate of the transformed
program, using concentration bounds to achieve a more precise bound.

Finally, in some applications the number of iterations
performed by the perforated loop is probabilistic; for instance, the program
\begin{center}
\begin{algorithmic}
  \State $\irnd{k}{\mathsf{factors}(n)}$;~ $s \gets 0$;
  \For {($\imath \gets 0, \imath < n, i \gets i + k$)}
  \State $\irnd{x}{\mu}$;~ $s \gets s + x$;
  \EndFor
  \State $s \gets k \cdot s$;
    \end{algorithmic}
\end{center}
selects uniformly at random a factor $k$ of the original number $n$ of
iterations, and performs $n/k$ iterations. It is possible to relate
the original and the perforated loop, using the [\textsc{While}] rule
as before.

\begin{figure}
    Original program:
{{ \begin{algorithmic}
  \State $s_1 \gets 0$;
  \For {($\imath_1 \gets 0, \imath_1 < 2 \cdot n, \imath_1 \gets \imath_1 + 1$)}
  \State $\irnd{x_1}{\mu}$;~ $s_1 \gets s_1 + x_1$;
  \EndFor
    \end{algorithmic}
    }}
    Perforated program:
{{ \begin{algorithmic}
  \State $s_2 \gets 0$;
  \For {($\imath_2 \gets 0, \imath_2 < n, \imath_2 \gets \imath_2 + 1$)}
  \State $\irnd{x_2}{\mu}$;~ $s_2 \gets s_2 + x_2$;
  \EndFor
  \State $s_2 \gets 2 \cdot s_2$;
    \end{algorithmic}
    }}
    Product program:
{{ \begin{algorithmic}
  \State $s_1 \gets 0$;~ $s_2 \gets 0$;~
         $\imath_1 \gets 0$;~ $\imath_2 \gets 0$;
  \While{$\imath_1 < 2 \cdot n$}
  \State $\irnd{x_1, x_2}{\CpId{\mu}}$;
  \State $s_1 \gets s_1 + x_1$; $s_2 \gets s_2 + x_2$;
  \State $\imath_1 \gets \imath_1 + 1$;
  \State $\irnd{x_1}{\mu}$;~ $s_1 \gets s_1 + x_1$;
  \State $\imath_1 \gets \imath_1 + 1$; $\imath_2 \gets \imath_2 + 1$;
  \EndWhile
  \State $s_2 \gets 2 \cdot s_2$;
    \end{algorithmic}
    }}
  \caption{Loop perforation}\label{fig:loopperf}
 \end{figure}

\subsection{Other Optimizations and Program Transformations}
\citet{BartheCK11} define an inductive method for building valid
product programs, and use their method for validating instances of
loop optimizations. Their method combines a rule for each program
construction and a rule akin to our [\textsc{Struct}] rule. Despite
this similarity, the two methods are fundamentally different: their
treatment of $\kwhile$ loops is restricted to synchronized executions.
As a consequence, their structural rule is based on a more advanced
refinement relation between programs. Nevertheless, we can reproduce
all their examples in our formalism, taking advantage of our more
powerful rule for loops.

\section{Related Work}
Relational logics can be seen as a proof-theoretical counterpart of
semantics-based relational methods such as logical relations. Under
this view, our logic bears strong similarities with proof-relevant
logical relations~\citep{BentonHN13}. As for proof-relevant logical
relations, we expect that manipulating explicit witnesses rather than
existentially quantified can help developing meta-theoretical studies
of our logic.

Much of the recent work on product programs and relational logics has
been motivated by applications to security and compiler correctness.
For instance, \citet{BartheDR04} explore self-composition for a
variety of programming languages and show that it induces a sound and
complete reduction of an information flow policy to a safety property.
Independently, \citet{DHS05:spc} consider self-composition and
deductive verification based on dynamic logic, also for verifying
information flow policies. Later, \citet{TerauchiA05} introduce the
class of 2-safety properties and show a reduction from 2-safety to
safety of the self-composed program. Their reduction is more efficient
than self-composition as it selectively applies self-composition or a
synchronous product construction akin to cross-products (described
below). Further improvements appear in~\citet{KovacsSF13,MullerKS15}.
In a related work,~\citet{BeringerH07} observe that one can encode
2-safety properties in standard Hoare logic, provided that assertions
are sufficiently expressive to model the denotational semantics of
programs. \citet{Beringer11} further refines this approach, by
introducing the notion of relational decomposition.

\citet{ZaksP08} define a cross-product construction, which is
well-suited for reasoning about programs with identical
control-flow. \citet{BartheCK11} generalize the notion of
cross-product by proposing a more general notion of product program
which subsumes self-composition and cross-products, and show how it
enables validation of common loop optimizations. Specifically, they
define an inductive relation for proving that $c$ is a valid product
for $c_1$ and $c_2$; informally, their rules closely follow those of
our system (for deterministic constructs), except for the general rule
for $\kwhile$ loops; instead, they use a rule that is closer to the
\Sprhl\ rule, and a rule akin to our [\textsc{Struct}] rule, with a
much stronger relationship between programs in order to compensate for
the lack of generality of their rule.  \citet{BartheCK16} carry a more
precise study of the relative expressiveness of product program
constructions and relational program logics. In a different thread of
work, \citet{BartheCK13} generalize the notion of product program so
that it supports verification of refinement properties (modeled by
universal quantification over runs of the first program and
existential quantification over runs of the second program), as well
as the 2-safety properties (modeled by universal quantification over
runs of the first and the second programs). These constructions are
focused on non-probabilistic programs. Motivated by applications to
differential privacy, \citet{BGGHKS14} define a specialized
product construction from probabilistic programs to deterministic
programs, so that the original program is differentially private,
provided its deterministic product program satisfies some Hoare
specification. To the best of our knowledge, this is the sole product
construction that goes from a probabilistic language to a deterministic one.

\citet{Benton04} and \citet{Yang07} were among the first to consider
relational program logics that support direct reasoning about two
programs. \citet{Benton04} introduces Relational Hoare Logic, proves
correctness of several program transformations, and soundly embeds
a type system for information flow security
into his logic. \citet{Yang07} defines Relational Separation Logic and
proves the equivalence between Depth-First Search and the Schorr-Waite
algorithm. \citet{BartheGZ09} develop probabilistic Relational Hoare
Logic, and use it for proving computational security of cryptographic
constructions. In a follow-up work, \citet{BartheKOZ12} develop an
approximate variant of probabilistic Relational Hoare Logic, and
verify differential privacy of several algorithms. More recently,
\citet{SousaD16} propose Cartesian Hoare Logic, an extension of
relational Hoare logic to an arbitrary finite number of executions.

Several authors have considered relational logics for higher-order
programs. \citet{NanevskiBG11} develop a relational logic to reason
about information flow properties of a higher-order language with
mutable state. \citet{GhaniFS16} introduce a relational type theory,
and a supporting categorical model, for reasoning about parametricity.
\citet{BartheFGSSZ14} propose a relational extension of a subset of the
$\mathrm{F}^*$; in a follow-up work, \citet{BartheGGHRS15} combine a
relational refinement type system with a graded monad which they use
for modeling differentially private computations.

There are several works that develop more specialized
program logics for analyzing relational properties of programs. For
instance, \citet{ABB06:popl} introduce independence assertions and a
supporting program logic for proving information flow security. In a
similar way, \citet{ChaudhuriGL10} propose a logical approach for
proving continuity properties of programs, and \citet{CarbinKMR12}
develop a logical approach for reasoning about the reliability of
approximate computation.

Further afield, there has been a significant amount of work on semantical
methods for probabilistic programs and processes initiated by \citet{Kozen79},
see e.g.~\citep{BizjakB15,DallagoSA14,SangiorgiV16,Kozen16} for some recent
developments. In conjunction with these semantics, research in deductive
verification methods for non-relational properties of probabilistic programs is
an active area of research; examples
include~\citep{Hartog:thesis,Kozen85,Morgan96,Ramshaw79} to cite only a few
systems.

\section{Conclusion and Future Directions}
We have introduced \Sprprhl, a new program logic that deepens the
connection between probabilistic couplings and relational verification
of probabilistic programs in two different ways. First, \Sprprhl
broadens the class of couplings supported by relational verification.
Second, \Sprprhl derivations explicitly build a probabilistic product
program, which can be used to analyze mixing times. We have shown the
flexibility of our approach on several examples.

There is ample room for future work. On the theoretical side, it would
be interesting to extend \Sprprhl to handle continuous distributions
as was recently done by \citet{sato2016approximate} for \Saprhl. Also,
we believe that we are just scratching the surface of probabilistic
product programs; there should be many further applications, notably
in relationship with path couplings, in domains such as Brownian
motion~\citep{Lindvall02}, molecular
evolution~\citep{dixit2012finite}, and
anonymity~\citep{GomulkiewiczKK03}. On the more practical side, it
would be natural to integrate \Sprprhl\ in
\EasyCrypt~\citep{BartheDGKSS13}, a proof-assistant used for reasoning
about computational security of cryptographic constructions.  We
expect that several proofs of cryptographic constructions can be
simplified using the new loop rule, and we are also planning to use
the loop rule in an ongoing formalization of indifferentiability of
the SHA3 standard for hash functions.

\paragraph*{Acknowledgments}

This work benefited from discussions with Thomas Espitau and Ichiro Hasuo. We
also thank the anonymous reviewers for their detailed comments, which improved
earlier versions of this work. This work was partially supported by NSF grants
TC-1065060 and TWC-1513694, and a grant from the Simons Foundation ($\#360368$
to Justin Hsu).

\bibliographystyle{abbrvnat}
\bibliography{header,refs}
\end{document}

%% file: prelude.tex
\usepackage{amsmath,amsfonts,amssymb,mathtools,stmaryrd}
\usepackage[bb=boondox]{mathalfa}
\usepackage{mathpartir}
\usepackage{url,graphicx,xspace,xcolor}
\usepackage{paralist,enumitem}
\usepackage{nicefrac}

\allowdisplaybreaks

\usepackage{todonotes}

\usepackage{hyperref}
\usepackage[capitalise]{cleveref}

\crefname{section}{\S}{\S}
\Crefname{section}{\S}{\S}

\crefname{prop}{proposition}{propositions}
\Crefname{prop}{Proposition}{Propositions}

\crefname{lem}{lemma}{lemmas}
\Crefname{lem}{Lemma}{Lemmas}

\crefname{thm}{theorem}{theorems}
\Crefname{thm}{Theorem}{Theorems}

\crefname{definition}{definition}{definitions}
\Crefname{definition}{Definition}{Definitions}

\usepackage{amsthm}

\newtheorem{theorem}{Theorem}
\newtheorem{proposition}[theorem]{Proposition}
\newtheorem{lemma}[theorem]{Lemma}
\theoremstyle{definition}
\newtheorem{definition}[theorem]{Definition}
\newtheorem*{remark*}{Remark}

\usepackage[noend]{algpseudocode}

\newcommand{\SIf}[3]{\textbf{if}~{#1}~\textbf{then}~{#2}~\textbf{else}~{#3}}
\newcommand{\SIfT}[2]{\textbf{if}~{#1}~\textbf{then}~{#2}}

\newcommand{\rname}[1]{[\textsc{#1}]}

\let\xinferrule\inferrule
\renewcommand{\inferrule}[3][]{%
  \ifx&#1&
    \xinferrule*{#2}{#3}
  \else
    \xinferrule*[right=\rname{#1}]{#2}{#3}
  \fi}

\newcommand{\EasyCrypt}{\textsf{EasyCrypt}\xspace}
\newcommand{\ssreflect}{\textsf{Ssreflect}\xspace}
\newcommand{\coq}{\textsf{Coq}\xspace}

\newcommand{\Sprhl}{\textsf{pRHL}\xspace}
\newcommand{\Sprprhl}{${\times}\mathsf{pRHL}$\xspace}
\newcommand{\Saprhl}{\textsf{apRHL}\xspace}

\newenvironment{tightcenter}{%
  \setlength\topsep{3pt}
  \setlength\parskip{3pt}
  \begin{center}
}{%
  \end{center}
}

\makeatletter
\newcommand{\subalign}[1]{%
  \vcenter{%
    \Let@ \restore@math@cr \default@tag
    \baselineskip\fontdimen10 \scriptfont\tw@
    \advance\baselineskip\fontdimen12 \scriptfont\tw@
    \lineskip\thr@@\fontdimen8 \scriptfont\thr@@
    \lineskiplimit\lineskip
    \ialign{\hfil$\m@th\scriptstyle##$&$\m@th\scriptstyle{}##$\crcr
      #1\crcr
    }%
  }
}
\makeatother

\newcommand{\NN}{\mathbb{N}}
\newcommand{\ZZ}{\mathbb{Z}}

\newcommand{\intv}[2]{[{#1},{#2}]}

\newcommand{\Dist}{\ensuremath{\mathbf{Distr}}}

\newcommand{\dnull}[1][]{{{\mathbb{0}}^{#1}}}
\newcommand{\dunit}[2][]{{{\mathbb{1}}^{#1}_{#2}}}

\newcommand{\dlet}[3]{\ExpD_{{#1} \sim {#2}} [{#3}]}

\newcommand{\glauber}{\mathsf{glauber}}
\newcommand{\rwalk}{\mathsf{rwalk}}
\newcommand{\dynkin}{\mathsf{dynkin}}

\newcommand{\chlg}{\mathsf{chlg}}
\newcommand{\ab}{\mathsf{ab}}

\newcommand{\Exp}{\mathbb{E}}
\newcommand{\ExpD}{\mathbb{E}}
\newcommand{\PrS}{{\textstyle\Pr}}

\newcommand{\E}[3]{\Exp_{{#1} \sim {#2}} [{#3}]}

\newcommand{\sE}[2]{\Exp_{{#1}} [{#2}]}

\renewcommand{\P}[3]{\PrS_{{#1} \sim {#2}} [{#3}]}
\newcommand{\sP}[2]{\PrS_{#1} [{#2}]}
\newcommand{\eP}[1]{\PrS[{#1}]}

\DeclareMathOperator{\id}{id}
\DeclareMathOperator{\opp}{opp}

\DeclareMathOperator{\supp}{supp}

\DeclareMathOperator{\TV}{TV}

\newcommand{\wt}[1]{|{#1}|}

\newcommand{\proj}{\pi}

\newcommand{\True}{{\ensuremath{\mathop{\top}}}}
\newcommand{\False}{{\ensuremath{\mathop{\perp}}}}

\newcommand{\Mem}{\mathsf{State}}

\newcommand{\OState}{\Omega}
\newcommand{\Var}{\mathcal{X}}

\newcommand{\pre}{\Phi}
\newcommand{\post}{\Psi}

\newcommand{\Vars}{\ensuremath{\mathbf{Vars}}}

\newcommand{\kw}[1]{{\ensuremath{\mathbf{#1}}}\xspace}

\newcommand{\kskip}{\kw{skip}}
\newcommand{\kabort}{\kw{abort}}
\newcommand{\kif}{\kw{if}}
\newcommand{\kthen}{\kw{then}}
\newcommand{\kelse}{\kw{else}}
\newcommand{\kwhile}{\kw{while}}
\newcommand{\kdo}{\kw{do}}

\newcommand{\kinc}{{\scriptstyle +\! +}}

\newcommand{\rndarrow}{%
  \stackrel{\raisebox{-.25ex}[.25ex]%
   {\tiny $\mathdollar$}}{\raisebox{-.25ex}[.25ex]{$\leftarrow$}}}

\newcommand{\iseq}[2]{{#1}; {#2}}
\newcommand{\iabort}{\kabort}
\newcommand{\iskip}{\kskip}
\newcommand{\iass}[2]{{#1} \leftarrow {#2}}
\newcommand{\irnd}[2]{{#1} \rndarrow {#2}}
\newcommand{\icond}[3]{\kif\ {#1}\ \kthen\ {#2}\ \kelse\ {#3}}
\newcommand{\icondT}[2]{\kif\ {#1}\ \kthen\ {#2}}

\newcommand{\iwhile}[2]{\kwhile\ {#1}\ \kdo\ {#2}}

\newcommand{\cfor}[2]{#1^{#2}}
\newcommand{\itfor}[2]{{#1}_{[#2]}}

\newcommand{\Expr}{\mathcal{E}}
\newcommand{\DExpr}{\mathcal{D}}
\newcommand{\cmd}{{\mathcal{C}}}

\newcommand{\subst}[2]{[#1 := #2]}

\newcommand{\pos}{x}
\newcommand{\start}{s}
\newcommand{\jump}{r}
\newcommand{\psum}{\mathsf{psum}}
\newcommand{\ssum}{\mathsf{sum}}
\newcommand{\rev}{\mathsf{rev}}

\newcommand{\coupled}[3]{%
  {#1} \blacktriangleleft \langle {#2} \mathrel{\&} {#3} \rangle}

\newcommand{\coupledsupp}[4]{%
  {#1} \blacktriangleleft_{#4} \langle {#2} \mathrel{\&} {#3} \rangle}

\newcommand{\cterminate}[3]{{#1} \Downarrow {\iwhile {#2} {#3}}}

\newcommand{\sem}[2]{{\llbracket {#1} \rrbracket}_{#2}}

\newcommand{\degueu}{\mbox{$\begin{array}{@{}c@{}} \, \\  \, \end{array}$}}
  
\newcommand{\Equiv}[4]{\left\{ \degueu #3 \right\}
             \mbox{\begin{tabular}{c} ${{#1}}$ \\ ${{#2}}$ \end{tabular}}
             \left \{ \degueu #4 \right\}}
 
\newcommand{\Equivw}[5]
           {\left\{ \degueu #3 \right\}
             \mbox{$\begin{array}{c} {{#1}} \\ {{#2}} \end{array}$}
             \left \{ \degueu #4 \right\} \leadsto {{#5}}}

\newcommand{\eqsem}[3]{#1 \vdash #2 \equiv #3}

\newcommand{\hoare}[3]{\{ #1 \} \; {#2} \; \{ #3 \}}

\newcommand{\CpDist}[2]{\mathcal{D}_{#1}(#2)}
\newcommand{\CpId}[1]{\mathcal{D}^{=}(#1)}

\DeclareMathOperator{\NEIGHBORS}{\mathcal{N}}
\newcommand{\neighbors}[1]{\NEIGHBORS_{#1}}

\DeclareMathOperator{\VALID}{\mathcal{V}}
\newcommand{\valid}[1]{\VALID_{#1}}

\DeclareMathOperator{\SAFE}{\mathcal{S}}
\newcommand{\safe}[1]{\SAFE_{#1}}
